\documentclass[runningheads,a4paper]{llncs}

\usepackage{amsmath}
\usepackage{amssymb}
\usepackage{multirow}
\usepackage[table]{xcolor}
\usepackage[ruled, vlined]{algorithm2e}

\usepackage{enumitem}
\usepackage{comment}

\makeatletter
\newcommand*{\rom}[1]{\expandafter\@slowromancap\romannumeral #1@}

\DeclareMathOperator*{\E}{\mathbb{E}}

\DeclareMathOperator*{\argmax}{arg\,max}

\makeatother

\begin{document}

\title{A Note on Relaxation and Rounding in Algorithmic Mechanism Design}

\author{Salman Fadaei}

\authorrunning{Fadaei} 
%
\tocauthor{Salman Fadaei}

\institute{
\email{salman.fadaei@gmail.com}}
\maketitle

\begin{abstract}
In this note, we revisit the \emph{relaxation and rounding} technique employed several times in algorithmic mechanism design.
We try to introduce a general framework which covers the most significant algorithms in mechanism design that use the relaxation and rounding technique.
We believe that this framework is not only a generalization of the existing algorithms but also can be leveraged for further results in algorithmic mechanism design.
Before presenting the framework, we briefly define algorithmic mechanism design, its connections to game theory and computer science, and the challenges in the field.\footnote{This work was done while the author was a graduate student in the Department of Informatics, TU M\"unchen, Munich, Germany.}

\keywords{Algorithmic Mechanism Design, Relaxation, Rounding}

\end{abstract}

\section{Introduction}

A set of numbers $a_1, a_2, \ldots, a_n$ is given. 
The algorithm below finds the maximum among these numbers.

\RestyleAlgo{plain}
\SetAlgoNoEnd
\begin{algorithm} 
\SetAlgoNoLine
\NoCaptionOfAlgo



	 $\max:=a_1$\;
	
 \For{$i\leftarrow 2$ \KwTo $n$}
 {
	\uIf {$\max < a_i$} {
		$\max := a_i$\;
	}

}

\caption{}
\label{alg:max-num}
\end{algorithm}

\RestyleAlgo{boxruled}

Now, consider a setting in which these numbers belong to self-interested agents.
Self-interested agents may misreport the numbers because of their own interests.
In this situation, the algorithm above  will be useful only if we certify that the true numbers are reported by the agents.
To this end, we need to model and understand the interests of the agents.

The dominant approach to model an agent's interests or \emph{preferences} is the utility theory.
A \emph{utility function} quantifies an agent's degree of preferences across a set of alternatives or outcomes.
A utility function assigns a real number to an outcome which defines the level of happiness of an agent with the outcome. 
Working with the idea of utility functions, in lieu of preferences, is pervasive and without loss of generality according to the von Neumann--Morgenstern theorem. 
Given a set of outcomes, the preference relationship of an agent identifies an ordering over the outcomes or lotteries of outcomes.
A lottery is a probability distribution over the outcomes.
The von Neumann--Morgenstern theorem states that for any preference relationship that satisfies a set of reasonable properties, there exists a utility function that correctly represents the preference relationship \cite{shoham2008multiagent}.
The von Neumann--Morgenstern theorem thus justifies the pervasive claim that a single-dimensional function (a utility function) suffices to describe preferences over an arbitrarily complicated set of alternatives.
We benefit from this observation and work with utility functions.

Finding the maximum among a set of numbers is an algorithmic problem.
If we assume each number is the \emph{private} information of an agent, then we also want the agents to disclose their numbers truthfully to the algorithm above.
Otherwise, the algorithm will be of no use.
Hence, we face a problem of \emph{mechanism design} rather than just algorithm design.
Mechanism design or \emph{implementation theory} is sometimes called ``inverse game theory".
Thus, to better understand mechanism design, we start with some game-theoretic fundamentals.

\section{Game Theory}
Game theory is the mathematical study of interaction among independent self-interested agents or players of a game. 
Normal-form games are the most fundamental representation of strategic settings.
In normal-form games, agents' moves are simultaneous.
We briefly introduce normal-form games which will help us define useful concepts such as equilibria.

\begin{definition}[Normal-form Games]
A (finite, n-person) normal-form game is a tuple $(\mathcal{I}, A, u)$, where:
\begin{itemize}

\item $\mathcal{I}$ is a finite set of n players, indexed by $i$
\item $A = A_1 \times \ldots \times A_n$ , where $A_i$ is a finite set of actions available to player $i$. Each vector $a = (a_1, \ldots , a_n) \in A$  is called an action profile.
\item $u = (u_1, . . . , u_n)$ where $u_i: A \mapsto \mathbb{R}$ is a real-valued utility (or payoff) function for player $i$.

\end{itemize}
\end{definition}

Agents choose strategies to play in a game. 
A strategy can be to select an action and play it which is called a \emph{pure strategy}.
Agents can also randomize over the set of available actions according to some probability distribution.
Such a strategy is called a \emph{mixed strategy}.
The set of all strategies for agent $i$ is denoted by $S_i$.
We use $s=(s_1,s_2,\ldots,s_n)$ and $s_{-i}$ to denote a strategy profile of the agents, and a strategy profile of the agents other than $i$, respectively.
Let $S_{-i}$ denote the set of all $s_{-i}$'s.

We analyze games using \emph{solution concepts}. 
Solution concepts are principles according to which we identify interesting subsets of the outcomes of a game.
Solution concepts can be seen as a way to predict the outcome of a game assuming some behavior of rational agents who try to maximize their own utilities.
The most significant solution concept in game theory is the \emph{Nash equilibrium}.
In the Nash equilibrium, we analyze the games from an individual agent's point of view.
The strategy that an agent chooses in a Nash equilibrium is called the \emph{best response}.

\begin{definition}[Best Response]
Agent $i$'s best response to the strategy profile $s_{-i}$ is a mixed strategy $s_i^* \in S_i$ such that $u_i(s_i^*,s_{-i})\geq u_i(s_i,s_{-i})$ for all strategies $s_i \in S_i$.
\end{definition}

Subsequently, the Nash equilibrium is defined as follows.

\begin{definition}[Nash Equilibrium]
A strategy profile $s=(s_1,\ldots,s_n)$ is a Nash equilibrium if, for all agents $i$, $s_i$ is a best response to $s_{-i}$.
\end{definition}

The other strategy that is of central significance is the \emph{dominant strategy}.
To define this strategy, first we define weak domination.
\begin{definition}[Weak Domination]
Let $s_i$ and $s'_i$ be two strategies of agent $i$. 
Then $s_i$ weakly dominates $s'_i$ if for all $s_{-i}$, $u_i(s_i,s_{-i}) \geq u_i(s'_i,s_{-i})$, and for at least one $s_{-i} \in S_{-i}$, it is the case that $u_i(s_i, s_{-i}) > u_i(s'_i,s_{-i})$.
\end{definition}

Consequently, a weakly dominant strategy is defined as follows.
\begin{definition}[Weakly Dominant Strategy]
A strategy is weakly dominant for an agent if it weakly dominates any other strategy for that agent.
\end{definition}

A strategy profile $(s_1,\ldots,s_n)$ in which every $s_i$ is weakly dominant for agent $i$ is called \emph{equilibrium in weakly dominant strategies}.
Other types of dominant strategies, strictly and very weakly dominant strategies, can be defined accordingly.
However, in this thesis, we only work with weakly dominant strategies and, for simplicity, we drop the modifier ``weakly".
An equilibrium in dominant strategy is a stronger concept than the Nash equilibrium.
Every equilibrium in dominant strategies is a Nash equilibrium but not always a Nash equilibrium is an equilibrium in dominant strategies.

Back to our discussion about mechanism design, let us recall that mechanism design is in fact inverse game theory.
In mechanism design, we assume unknown individual preferences (utilities), and attempt to design a game such that no matter what the unknown preferences are, the equilibrium (e.g. equilibrium in dominant strategies) of the game is guaranteed to have a certain set of properties. 
Mechanism design is perhaps the most ``computation-driven" part of game theory since it concerns itself with the design of protocols for distributed decision makers.
However, since the decision makers are not necessarily cooperative, one can think of mechanism design as an exercise in ``incentive engineering".

\section{Mechanism Design for Combinatorial Auctions}

While mechanism design is a broad area in game theory, we study mechanism design for welfare maximization in combinatorial auctions.
The problem of maximizing \emph{social welfare} in \emph{combinatorial auctions} (CAs) is of central importance in both theory and practice of mechanism design.
In CA, there are $m$ items for sale and $n$ bidders competing for these items.
Each bidder $i$ has a private valuation $v_i(S)$ for each package $S$ of items.
Alternatively, we say $v_i$ is the private type of bidder $i$.
The social welfare of an allocation $S_1, \ldots, S_n$ of items to the bidders is $\sum_{i=1}^n v_i(S_i)$.

An outcome, in the setting of CA, is a pair of allocations as well as the payments to the bidders: $o=(S_1,\ldots,S_n, p_1, \ldots, p_n)$.
In CA, bidders usually have quasi-linear utility functions.
A Quasi-linear utility function is of the form $u_i(S_i)=v_i(S_i)-p_i$: the utility of a bidder $i$ for an allocation is the private valuation of the bidder less the payment for the allocation.
With this type of utility functions, the mechanism can charge or reward the bidders by an arbitrary amount of money in order to incentivize the bidders to truthfully report their values. 

A mechanism implements a \emph{social choice function} in an \emph{equilibrium}.
A social choice function is a mapping from types to outcomes.
In our setting, the social choice function is optimizing social welfare: $\max_{S_1,\ldots,S_n}\sum_{i=1}^n v_i(S_i)$.

Bidders choose strategies to bid in an auction. 
A strategy is a function from private types (valuations) to actions (bids).
Bidders are rational, that is they bid to maximize their own utility.
A strategy is dominant, if the bidder bids regardless of the bids submitted by other bidders.
A strategy profile in which every bidder has a dominant strategy is called an \emph{equilibrium in dominant strategies}.
We seek a mechanism that maximizes social welfare in an equilibrium in dominant strategies.

In order to address a mechanism design problem, the \emph{revelation principle} encourages us to focus on \emph{incentive-compatible} \emph{direct} mechanisms.
In a direct mechanism, the only action available to each agent is to report its private information.
A direct mechanism is incentive compatible or \emph{truthful} in dominant strategies if every agent $i$ has a dominant strategy to truthfully report its private type.
\begin{definition}[Revelation Principle]
If there exists any mechanism that implements a social choice function in dominant strategies then there exists a direct mechanism that implements the social choice function in dominant strategies and is truthful.
\end{definition}

In other words, the task of solving a mechanism design problem can be addressed by finding a direct and truthful mechanism.
Thus, we limit ourselves to social-welfare maximizing mechanisms which are direct and incentive compatible in dominant strategies (truthful).
In the present text, we use the terms ''truthfulness" and ''incentive compatibility" interchangeably, and by either we mean ''incentive compatibility in dominant strategies".
In order to use this useful observation, we have to take into consideration two important points.
First, by using the revelation principle the computational burden of finding strategies for the agents is pushed onto the mechanism.
This fact draws our attention to the computational complexity of the underlying algorithm of the mechanism.
Second, since agents have to fully reveal their types, this can place a burden on the communication channel.
In our mechanisms, we have considered these points, and the proposed mechanisms do not suffer from any of these issues.
Given that we only work with \emph{direct-revelation} mechanisms, from now on, we use the word mechanism rather than direct mechanism.

To be more specific about mechanism design for the CA setting, we say a mechanism is a protocol with three components.
First, the mechanism extracts a report from each bidder that describes the bidder's valuations.
Second, the  mechanism computes an allocation that determines the package of items to be allocated to each bidder (allocation rule).
Third, the mechanism charges each bidder an amount of money (payment rule).

If each bidder, in anticipation of the outcomes of the allocation and payment rule, reports her valuation truthfully, it is said that the auction mechanism is incentive-compatible in dominant strategies or truthful.

Incentive compatibility in dominant strategies is stronger than the \emph{Bayes-Nash incentive compatibility}.
In Bayes-Nash incentive compatible mechanism, a bidder is truthful provided that other bidders bid truthfully.
Incentive compatibility in dominant strategies assures that the strategy of each bidder, regardless of the bids submitted by other bidders, is to bid truthfully.

Back to our example of finding the maximum among numbers, if we assume the agents have quasi-linear utilities, we can model the problem as a mechanism design problem for selling one item truthfully to the agent whose private value for the item is maximum.
With quasi-linear utility functions, Vickrey's mechanism called \emph{second price auction} solves this mechanism design problem.
In particular, it finds the maximum value reported by the agents and charges the corresponding agent an amount equal to the second highest number.
It is well known that in a second price auction, each agent truthfully reports its private type \cite{Vickrey61}.

Thus, the algorithm for finding the maximum works correctly provided that the agents have quasi-linear utilities and are told about the  possible payment by the owner of the maximum number.
This phenomena is quite amazing in that despite the private data and pure selfish behavior, the maximum number can be correctly found.
All the field of mechanism design is just a generalization of this possibility.

\section{VCG Auctions}


In combinatorial auctions, the goal of mechanism design is to maximize social welfare while providing the bidders with incentives to truthfully report their private valuations. 
In economics, the celebrated Vickrey-Clarke-Groves (VCG) mechanism provides both objectives.
In particular, VCG finds the welfare maximizing allocation $S_1^*, \ldots, S_n^*$, and employs the following payment rule in order to guarantee incentive compatibility.
Each bidder $i$ pays her externality, the amount by which
his allocated bundle reduced the total reported value of the bundles allocated to others:
$\sum_{l=1, l\neq i}^n v_i(S'_i)-\sum_{l=1, l \neq i}^n v_i(S_i^*)$.
The first term is the maximum social welfare obtainable without bidder $i$; that is $(S'_1,\ldots,S'_n)$ maximizes social welfare in a market without bidder $i$. 
The second term is the welfare of the market without bidder $i$ under the social welfare maximizing allocation  $S_1^*, \ldots, S_n^*$.

The success of VCG strongly relies on optimizing the social welfare.
Computationally, this is not always possible because CA is a NP-hard problem.
Even a restricted type of CAs in which bidders are single-minded is NP-hard since we can show a reduction from the independent-set problem to the CA problem with single-minded bidders \cite{blumrosen2007combinatorial}.
A single-minded bidder is interested only in a single specific package of items and gets a specific value if she gets the whole package (or any superset) and zero value for any other package.
An essential question arises here.
Can we use approximation algorithms -- that are usually employed by computer scientists to tackle hard problems -- and then use VCG-like payments to obtain truthfulness?

\section{Truthful Approximation Mechanisms}
An algorithm is an $\alpha$-approximation algorithm (or has approximation ratio $\alpha$) if the the objective value of the computed solution is at least a factor $\alpha$ of the value of the optimal allocation (the allocation with maximum objective function value).
Unfortunately, approximating the social welfare and then using the idea of VCG payments does not preserve truthfulness \cite{Nisan01algorithmicmechanism}.
For instance, assume in a single-item auction, we assign the item to the second highest bid (as an approximation to the highest value) at the price of the third highest bid. 
This auction has no equilibrium because for example the bidder with the highest value would decrease her bid below the second bid to obtain the item, and similarly, the second bidder would decrease her bid and so on.

Resolving the tension between approximation algorithms and incentive compatibility is the topic of \emph{Algorithmic Mechanism Design} \cite{Nisan01algorithmicmechanism}.
In particular in algorithmic mechanism design we ask ourselves if we can theoretically guarantee a ratio for an incentive-compatible approximation algorithm with private input data that (almost) equals the ratio of and approximation algorithm presented for the same problem with public input data?
This problem has been extensively studied in the last two decades for various types of valuations \cite{blumrosen2007combinatorial}.

Not surprisingly, randomized algorithms have been proven to be more promising than deterministic algorithms in algorithmic mechanism design.
For example, the two main frameworks presented by Lavi and Swamy (LS framework) and Dughmi et al. (convex rounding) are randomized and yield truthfulness in expectation.
Truthfulness in expectation certifies that bidders who are risk-neutral (expected utility maximizer) have no incentive to report false valuations.

Generally speaking, the two frameworks rely on the idea of \emph{relaxation} and \emph{rounding}.
The LS framework first optimizes a LP relaxation of the problem, and then uses a specific rounding technique called \emph{meta-randomized rounding} originally proposed by Carr and Vempala \cite{carr2000randomized}.
Meta-randomized rounding yields a convex decomposition of a fractional point into polynomially-many integer points.
The rounding technique has been successfully applied to mechanism design with \cite{lavi2011truthful} and without money \cite{dughmi2010truthful}.
One drawback of the first presentations of the meta-randomized rounding is that, the technique relies on the ellipsoid method which is notoriously of low practical usability and is mostly of theoretical importance \cite{carr2000randomized,lavi2011truthful}.

The convex rounding technique optimizes a convex function which is the image of a rounding algorithm over a relaxed set of solutions.
Convex rounding optimizes over the outcome of a rounding algorithm, and this way the rounding algorithm is integrated into the optimization problem.
We will discuss both techniques in Section \ref{tie}.

\section{Approximations to Obtain Truthfulness}

Not always quasi-linear utilities and transfer of money are available.
There are many situations in which no money is involved because of the nature of the problem or because of the law.
Truthful mechanism design without money under general preferences is a classic topic in social choice theory.
Without money, mechanism designers face significant obstacles in truthful mechanism design.
Think about the combinatorial auctions setting defined above.
The main idea of truthfulness strongly relies on the payments.
Payments enter the utility of the bidders and are leveraged to incentivize the bidders to bid truthfully.
When payments are not available, the utility of the agents depend only on their valuation for the allocated package of items.
For payment-free environments, it sounds unlikely for the mechanism designers to be able to provide agents with incentives to truthfully report their private types.

It is in fact the case; the Gibbard-Satterthwaite theorem proves that the class of truthful mechanisms is limited to dictatorships \cite{Gibbard73,Satterthwaite75}. 
In particular, it states that any truthful social choice function which selects an outcome among three or more alternatives has to be trivially aligned with the preference of a single agent.
There have been a number of extensions analyzing more specific domains without money, all resulting in impossibility results \cite{Papai01,Ehlers03,Hatfield09}. 
To circumvent the Gibbard-Satterthwaite impossibility, researchers have introduced restricted domains with
additional assumptions to admit truthful mechanisms.
For example, when agents valuations are restricted to single-peaked preferences over a one-dimensional public space, returning the \textit{median} of the peaks determines a truthful social choice \cite{moulin1980strategy}.
Single-peaked preferences have a single most-preferred point in an interval, and are decreasing as one moves away from that peak.

Procaccia and Tennenholtz introduced the technique of welfare approximation as a means to drive truthful approximation mechanisms without money \cite{procaccia2013approximate}. This type of approximation is not meant to handle computational intractability but a method to achieve truthfulness. 
We show applications of this technique in mechanism design without money for some restricted domains of preferences.

\section{Relaxation and Rounding}

Usually, when faced with NP-hard problems, computer scientists turn to approximations or heuristics.
An approximation algorithm runs in polynomial time in the size of the encoding of input data, and returns a provable approximation of the optimal solution.
The approximation ratio is proven with respect to the worst-case analysis of the algorithm.
An algorithm is an $\alpha$-approximation algorithm (or has approximation ratio $\alpha$) if the the objective value of the computed solution is at least a factor $\alpha$ of the value of the optimal solution.

Relaxation and Rounding (henceforth, \emph{relax} and \emph{round}) is a well-known and ubiquitous technique for designing  approximation algorithms.
The relax and round technique is also the most successful machinery to design \emph{truthful} approximation algorithms that run in polynomial time.
For quasi-linear utilities the idea of relax and round is omnipresent.
The idea has also been applied to mechanism design without money.

Given the significance of the idea of relax and round in the literature, and that several truthful mechanisms utilize the idea, in this note, we elaborate further on it.
We explain the idea for quasi-linear valuations first, and then we explain how to apply the technique to mechanism design without money.

\section{Setting}
In a market, a group of $n$ bidders are vying for a set of items.
A feasible solution is an allocation of items to bidders, satisfying (possibly) some given constraints.
Each bidder $i$ has a \emph{private} valuation function $v_i$ with value $v_i(S_i)$ for a feasible solution $(S_1,S_2,\ldots, S_n) \in \mathcal{S}$, where $\mathcal{S}$ denotes the set of all feasible solutions.
For simplicity, we denote the value of bidder $i$ for solution $x=(S_1,S_2,\ldots, S_n)$ by $v_i(x)$, where $v_i(x)=v_i(S_i)$.
Let $V_i$ denote the set of all valuation functions of bidder $i$, and $V=V_1\times V_2 \times \ldots \times V_n$ denote the set of all profiles of valuations. 
Denote the profile of valuations of bidders other than $i$ by $v_{-i}$.
Define function $f:\mathcal{S} \rightarrow \mathbb{R}_+$ with $f(x)=\sum_i v_i(x)$ for all $x \in \mathcal{S}$.
 
The following optimization problem expresses the social welfare maximization problem.
  \begin{align} \label{prog-max-int}
    \text{Maximize} \quad &  f(x)  \\
	 \text{subject to} \quad & x \in \mathcal{S} \notag
  \end{align}
The social welfare of a solution $x \in \mathcal{S}$ thus equals $f(x)$.
The problem that we study is to analyze Problem \ref{prog-max-int} as a mechanism design problem.
Specifically, we wish to maximize social welfare while making sure each bidder $i$ reports her valuation function $v_i$ truthfully.
Given that $v_i$'s are private, we build program \ref{prog-max-int} from the reported valuations (bids).
If the optimization Problem \ref{prog-max-int} can be solved in polynomial time, then VCG can be employed as a dominant-strategy incentive-compatible (truthful) mechanism.
However, in many cases including combinatorial auctions, Problem \ref{prog-max-int} is NP-hard, and we can only hope for approximations of the optimal solution.
In what follows, we describe how to employ the idea of relax and round to achieve mechanisms that satisfy the weaker notion of truthfulness in expectation.

\section{Relaxation}
The idea of relax and round proceeds as follows.
We find a relaxation of Problem \ref{prog-max-int}, solve it and then carefully round the outcome so as to obtain particular properties.
The relaxed problem is as follows.
  \begin{align} \label{prog-max-relaxed}
    \text{Maximize} \quad &  L(x)  \\
	 \text{subject to} \quad & x \in P \notag
  \end{align}

Where $P$ is a relaxed set of solutions that contains $\mathcal{S}$, and may contain infeasible solutions.
Polytope $P$ is the intersection of a set of linear constraints in the positive orthant, and is a \emph{packing} polytope: if $x \in P$ and $y \leq x$ then $y \in P$.
Function $L:P \rightarrow \mathbb{R}_+$ is concave or linear in $x$, and for every $x \in \mathcal{S}$, it is the case that $L(x) \geq \alpha f(x)$ for some $\alpha$, $0<\alpha\leq 1$.
As we see below, $\alpha$ will be the approximation ratio of the algorithm.
Function $L$ is separable, i.e. $L(x) = \sum_i L_i(x)$, however $L_i(x)$ need not  necessarily be concave.
Let us call $L$ \emph{relaxed objective function}.
Recall, the mechanism designer has only access to the reported valuations, thus all value computations are carried out with respect to the reported valuations. 

Given that function $L$ is concave and all constraints of polytope $P$ are linear, Problem \ref{prog-max-relaxed} is a convex optimization problem that can be solved efficiently.
Once we solved Problem \ref{prog-max-relaxed}, we proceed to the next step, rounding.

\section{Rounding}
A randomized rounding scheme $r: P \rightarrow \mathcal{S}$ returns a randomized solution $X \sim r(x)$ for any $x\in P$. 
The rounding scheme is \emph{oblivious}, i.e. it does not depend on the valuations.
Solution $X=(S_1,\ldots,S_n)$ is always feasible, i.e. $Pr[X \notin \mathcal{S}] = 0$.
With regard to the expected value of the rounded solution (with respect to reported valuations), one of the following cases may occur:
\begin{enumerate} [label={(\alph{enumi})}] 
\item \label{rnd-bigger} 
$\E[f(X)] > L(x)$. In this case, apply an oblivious rounding scheme $r'$ to $X$ in order to obtain $X'=(S'_1,\ldots,S'_n)$ such that $\E[f(X')]=L(x)$.

\item \label{rnd-smaller} 
$\E[f(X)] < L(x)$ but $\E[f(X)] \geq \beta L(x)$ for some $\beta$, $0<\beta\leq 1$. Apply an oblivious rounding $r'$ to $X$ in order to obtain $X'$ such that $\E[f(X')]=\beta L(x)$. In this case, the approximation ratio of the algorithm will be $\alpha \beta$.
\item \label{rnd-equal} 
$\E[f(X)]=L(x)$. No more action is needed, $r'(X)=X$.
\end{enumerate}

We will refer to the three cases of the rounding, \ref{rnd-bigger}, \ref{rnd-smaller}, and \ref{rnd-equal} in the following.

\section{Truthful-in-expectation Mechanism} \label{tie}
Here, we explain how to use the foregoing idea of relax and round to define a truthful-in-expectation mechanism.
A mechanism comprises an allocation rule $\mathcal{A}$ and a payment rule $p$.
A mechanism $(\mathcal{A},p)$ is truthful-in-expectation if for every player $i$ with true valuation $v_i$ and reported valuation $v'_i$, we have
\begin{equation} \label{equ-tie}
\E[v_i(\mathcal{A}(v))-p_i(v)] \geq \E[v_i(\mathcal{A}(v'_i,v_{-i}))-p_i(v'_i,v_{-i})].
\end{equation}
The expectation in (\ref{equ-tie}) is taken over the coin flips of the mechanism.
In order to achieve a truthful-in-expectation mechanism, we propose the following allocation rule.
\\
\begin{algorithm}[H]\label{alg-midr}

	1. Let $x^*=\argmax_{x\in P} L(x)$
	
	2. Let $X' \sim r'(r\small(x^*\small))$ \tcc*	 {\footnotesize $r'$ is chosen according to the occurring case: \ref{rnd-bigger}, \ref{rnd-smaller}, or \ref{rnd-equal}.}
	
	\Return $X'$
	
\caption{Allocation Rule.}
\end{algorithm}

To define the payments, let $x^*=\argmax_{x\in P} L(x)$, and $(S_1,\ldots,S_n) \sim r'(r\small(x^*\small))$. 
The payment rule determines the payment by each bidder. 
The expected VCG payment of any bidder $k$ is equal to 

\begin{equation} \label{payment-rule}
p_k = \max_{x \in P} L^{-k}(x) - \E[\sum_{i,i\neq k} v_i(S_i)].
\end{equation}

Where $L^{-k}(x)$ is the relaxed objective function for a market in which bidder $k$ is discarded, and the allocation rule is run for the society of all other bidders.
Equivalently, we can say $L^{-k}(x)$ is defined assuming $v_k(x)=0$, $\forall x \in P$.
Let $x_{-k}^*=\argmax_{x\in P} L^{-k}(x)$, and $(T_1,\ldots,T_n) \sim r'(r\small(x_{-k}^*\small))$. 
If we let any bidder $k$ pay an amount of $\sum_{i,i\neq k} v_i(T_i)-\sum_{i,i\neq k} v_i(S_i)$, then using linearity of expectation, the expected value of this payment is equal to the expression in (\ref{payment-rule}).


\begin{theorem} \label{thm-general-ic}
The allocation rule above (Algorithm \ref{alg-midr}) supplemented with the payment rule (\ref{payment-rule}), constitute a truthful-in-expectation mechanism with a provable approximation ratio of $\alpha$ (or $\alpha \beta$).
\end{theorem}
\begin{proof}
Let $x^*=\argmax_{x\in P} L(x)$. 
Fix bidder $k$ and $v_{-k}$ the valuations of bidders other than $k$.
Let $(S_1, \ldots, S_n) \sim r'(r(x^*))$.
Let $C=\max_{x \in P} L^{-k}(x)$.
The expected utility of bidder $k$ from reporting true values $v_k$ is the following.
\begin{displaymath}
\begin{array}{ll}
\E[u_k(v)]&=\E[v_k(S_k)]-p_k \\
&=\E[v_k(S_k)]-( C - \E[\sum_{i,i\neq k} v_i(S_i)]) \\
&=L(x^*)-C.
\end{array}
\end{displaymath}

Recall, $L(x^*)=\E[\sum_{i} v_i(S_i)]$ by the construction of the rounding.
The second term $C$ is independent of the valuations of bidder $k$, thus the bidder cannot influence it.
In order to increase her expected utility, the bidder has to increase the first term, i.e. $L(x^*)$, but the first term is already the maximum possible value.
Hence, the bidder cannot improve her expected utility.
To be more specific, let bidder $k$ report $v'_k$ rather than $v_k$.
Let $x'=\argmax_{x\in P} L'(x)$, where $L'$ is the relaxed objective function obtained from the new reported valuations $(v'_k, v_{-k})$. 
Let $(S'_1, \ldots, S'_n) \sim r'(r(x'))$.
The expected utility of the bidder, in this case, will be lower as shown below.

\begin{displaymath}
\begin{array}{ll}
\E[u_k(v'_k, v_{-k})]&=\E[v_k(S'_k)]-(C - \E[\sum_{i,i\neq k} v_i(S'_i)])  \\
&=\E[v_k(S'_k)]+\E[\sum_{i,i\neq k} v_i(S'_i)]- C \\
&=L(x')-C \\
&\leq L(x^*) - C \\
&=\E[u_k(v)]
\end{array}
\end{displaymath}

For the first equality, recall $L'(x')=\E[v'_k(S'_k)]+\E[\sum_{i, i\neq k}v_i(S'_i)]$, however, the utility of bidder $k$ is calculated with respect to true valuation $v_k$.
Thus, the bidder has no incentive to report a false valuation function.

The approximation ratio admits a simple proof.
Consider an optimal integral solution $Y^*=(S^*_1,\ldots,S^*_n)$.
Since $Y^* \in P$, we have $L(Y^*) \leq L(x^*)$, where $x^*=\argmax_{x\in P} L(x)$.
By definition of $L$, we have $L(Y^*)\geq \alpha f(Y^*)$. 
Let $X' \sim r'(r(x^*))$.
By the construction of the rounding, we have $\E[f(X')]=L(x^*) \geq L(Y^*) \geq \alpha f(Y^*)$, the desired conclusion. 
Similarly, we can show the approximation ratio of $\alpha \beta$ for case \ref{rnd-smaller} of the rounding.
This completes the proof.
\end{proof}

Therefore, we explained how one can use the relax and round technique to define an allocation and a payment rule in order to obtain a truthful-in-expectation mechanism.

The technique of relax and round that we explained here is an abstraction of the existing truthful approximation algorithms that employ relaxation and rounding technique.
Here, we abstract away from many details of the realization of the technique. 
These details  are varied for different applications.
For example, in our explanation, we gave no specification about how to define function $L$, or the rounding algorithms $r$ and $r'$.
It is instructive to note that, from the provided explanation, we observe that the relax and round technique has an extra step (calling $r'$) when used for designing \emph{truthful} approximation mechanisms compared to the usage of the technique for approximation algorithms.
This can also be seen as a reason for the fact that designing truthful approximation algorithms is more stringent than designing (untruthful) approximation algorithms.
Several truthful mechanisms that are grounded in mathematical optimization follow a variation of the foregoing relax and round technique.
We briefly explain this observation in the following.

The framework proposed by Lavi and Swamy (LS framework) can be described as follows.
The LS framework, first relaxes the underlying integer program of the problem to a linear program (LP), and solves the LP.
Then the solution of the LP is scaled down by a specific factor (the integrality gap of the underlying polytope), and the meta-randomized rounding is applied to the scaled-down solution \cite{lavi2011truthful}.
For more details about LS framework, we refer the reader to \cite{lavi2011truthful}.

In the language of the relax and round technique that we described above, the LP in the LS framework is closely related to Problem \ref{prog-max-relaxed}. 
Function $L$ is in fact the objective function of the LP, therefore $L$ is a linear function.
Polytope $P$ is the region of feasible solutions of the LP scaled down by the integrality gap.
Therefore, $P$ does not contain all solutions of $\mathcal{S}$.
However, $P$ has a special property: for any $v \in V$, there exists a point $x \in P$ such that $L(x) \geq \alpha\cdot \max_{y \in \mathcal{S}}\sum_i v_i(y)$, where $1/\alpha$ is the integrality gap of the LP. 
Moreover, in the LS framework for every $x \in \mathcal{S}$, $L(x) = f(x)$.
The meta-randomized rounding used in the LS framework results in case \ref{rnd-equal} of the rounding step.
Therefore, no second rounding $r'$ is required.

Convex rounding is the other general framework for designing truthful approximation mechanisms \cite{dughmi2011convex}.
In convex rounding, the allocation rule optimizes directly on the outcome of the rounding algorithm, rather than over the outcome of the relaxation algorithm.
The convex rounding can be described by the relax and round technique as follows.
Function $L$ is a concave function.
Polytope $P$ is simply a relaxed set of feasible solutions.
The rounding scheme ends in case \ref{rnd-equal} of the rounding step.

Archer et al. propose a truthful-in-expectation mechanism for combinatorial auctions with single parameter agents using the idea of relax and round \cite{archer2004approximate}.
The authors employ a linear programming relaxation for the problem, and the rounding case \ref{rnd-smaller} occurs.
Similarly, Dobzinski et al. propose a truthful-in-expectation mechanism for combinatorial auctions with subadditive bidders.
They also use a linear programming relaxation and the rounding case \ref{rnd-smaller} occurs \cite{dobzinski2010truthfulness}.

In the solution presented for the generalized assignment problem \cite{fadaei2014gap}, 
the authors use a concave objective function and a polytope which contains all integral solutions as well as infeasible solutions, and 
the rounding case \ref{rnd-bigger} happens.
In addition, the authors employ the following new observation.
From the explanation above, we know that the rounding schemes $r$ and $r'$ are oblivious.
It is possible to extend this idea to non-oblivious rounding schemes given that the following condition holds.
A rounding scheme $r:P\times V \rightarrow \mathcal{S}$ preserves truthfulness if for any misreported valuation $v'_i$, we have $\E[f(r(x,v'_i,v_{-i}))] \leq \E[f(r(x,v))]$.
This condition assures that the expected value of the social welfare is maximized by truthful bidding.

In the literature, the allocation rule defined in Algorithm \ref{alg-midr} is termed Maximal-In-Distributional-Range (MIDR) algorithm.
A MIDR algorithm optimizes over a distributional range of the solutions.
To preserve truthfulness, the distributional range of the solutions must be fixed independently of the valuations of the bidders \cite{Dobzinski09}.
The distributional range implied by the relax and round technique described above is the image of the rounding scheme:
\begin{equation} \label{midr-in-relaxation}
\textit{Distributional Range} \equiv \bigcup_{x \in P} \{r'(r(x))\}.
\end{equation}
The (distributional) range in (\ref{midr-in-relaxation}) is independent of bidders' valuations. 
Since we optimally maximize over this range, Algorithm \ref{alg-midr} is in fact a MIDR algorithm.

\section{Mechanism Design without Money}
To use the relax and round technique for mechanism design without money, we do as follows.
First, we relax the underlying integer program of the problem to a linear program.
Let $P$ denote the region of the feasible solutions to the linear program.
We find a fractional point in $P$ by using an algorithm that is \emph{fractionally truthful}.
A fractionally truthful algorithm $\mathcal{A}$ takes $v\in V$ and $P$ as input, and computes a point $x \in P$ such that for any bidder $i$ and untruthful valuation $v'_i$, we have $v_i(\mathcal{A}(v)) \geq v_i(\mathcal{A}(v'_i, v_{-i}))$.
To maintain a good approximation of the total value, the value of the computed point $x$ must be at least as good as $\alpha$ times ($0<\alpha\leq 1$) the optimal solution to the relaxed problem.

Next, we round $x$ using a randomized rounding scheme $r$ with the following properties.
If $X \sim r(\mathcal{A}(v,P))$ we must have
$(i)$ (feasibility) $Pr[X \in \mathcal{S}]=1$ and $(ii)$ (truthfulness)  $\forall i$, $\E[v_i(X)]=\beta v_i(x)$ for $0<\beta \leq 1$.

This technique results in a truthful-in-expectation $(\alpha \beta)$-approximation algorithm.
We refer the reader to \cite{dughmi2010truthful}, \cite{chen2013truthful} and \cite{fadaeiGapWoMoney2016}.

\bibliographystyle{splncs03}
\bibliography{literature}

\end{document}